\documentclass[a4paper,UKenglish,cleveref, autoref, thm-restate]{lipics-v2021}
\usepackage{booktabs}

\usepackage{amsmath}
\usepackage{amssymb}
\usepackage{amsfonts}
\usepackage{mathtools}

\usepackage[linesnumbered,ruled,vlined]{algorithm2e}
\usepackage{graphicx}
\usepackage{tikz}
\usetikzlibrary{arrows.meta, positioning, fit, backgrounds, calc}

\usepackage{pgfplots}
\pgfplotsset{compat=1.18}

\usepackage{caption}
\usepackage{subcaption}
\usepackage{booktabs}
\usepackage{xcolor}
\usepackage{xspace}
\usepackage{textcomp}
\usepackage{csquotes}
\usepackage{multicol}
\usetikzlibrary{decorations.pathreplacing}

\usepackage[utf8]{inputenc}
\usepackage[T1]{fontenc}
\usetikzlibrary{fit}
\usepackage{float}    % for better float control
\usepackage{array}

\newcommand{\bpqbc}{balanced $(3,3)$-biclique\xspace}
\newcommand{\pqbc}{$(p,q)$-biclique\xspace}

\newcommand{\pwedge}{$p$-\textsc{wedge}\xspace}

\newcommand{\stredge}{signed $3$-wedge\xspace}
\newcommand{\spwedge}{signed $p$-wedge\xspace}

\title{Counting Small Balanced $(p,q)$-bicliques in Signed Bipartite Graphs} %TODO Please add

\author{Mekala Kiran}{Bits Pilani Hyderabad Campus, [Hyderabad], India}{p20220017@hyderabad.bits-pilani.ac.in}{}{}

\author{Apurba Das}{Bits Pilani Hyderabad Campus, [Hyderabad], India}{apurba@hyderabad.bits-pilani.ac.in}{}{}

\author{Suman Banerjee}{Indian Institute of Technology Jammu, [Jammu \& Kashmir], India}{suman.banerjee@iitjammu.ac.in}{}{}

\author{Tathagata Ray}{Bits Pilani Hyderabad Campus, [Hyderabad], India}{rayt@hyderabad.bits-pilani.ac.in}{}{}
 %TODO mandatory: Please choose ACM 2012 classifications from https://dl.acm.org/ccs/ccs_flat.cfm 


\category{} %optional, e.g. invited paper

\nolinenumbers %uncomment to disable line numbering

\EventEditors{John Q. Open and Joan R. Access}
\EventNoEds{2}
\EventLongTitle{24th Symposium on Experimental Algorithms(SEA 2026)}
\EventShortTitle{SEA 2026}
\EventAcronym{SEA}
\EventYear{2026}
\EventDate{June 22–24, 2026}
\EventLocation{Copenhagen, Denmark}
\EventLogo{}
\SeriesVolume{42}
\ArticleNo{23}
\begin{document}

\maketitle

% \begin{abstract}
% Motif-based analysis is essential for understanding the structure of large-scale networks. In contrast to unsigned graphs, motif analysis in signed bipartite graphs has received limited attention. The smallest non-trivial motif in a signed bipartite graph is a signed butterfly, or $(2,2)$-biclique, which captures only local balance patterns and cannot reveal higher-order balance relationships. To overcome this limitation, we study balanced $(3,3)$-bicliques, the smallest dense extension beyond butterflies, which model global balance consistency in signed bipartite graphs. To the best of our knowledge, this is the first work on counting balanced $(3,3)$-bicliques in signed bipartite graphs, where balance is defined by the absence of unbalanced butterflies. We present a baseline adapted from biclique enumeration methods and analyse its limitations. We then propose two efficient algorithms: a wedge-centric approach (BBWC) that enforces balance constraints during enumeration and a more efficient triplet-centric approach (BBTC) that directly enumerates feasible vertex triplets. Extensive experiments on large real-world signed bipartite datasets demonstrate that the triplet-centric approach significantly outperforms the baseline, achieving speedups of over $100\times$ compared to SBCList++, an adaptation of the BCList++ algorithm.
% \end{abstract}

\begin{abstract}
% A bipartite graph is the best representation for many complicated networks in real-world applications. Motif-based analysis is essential for understanding the structure of large-scale networks, and bipartite graphs are no exception. In contrast to unsigned graphs, motif analysis in signed bipartite graphs has received limited attention. The smallest non-trivial motif in a signed bipartite graph is a signed butterfly, or $(2,2)$-biclique, which captures only local balance patterns and cannot reveal higher-order balance relationships. To overcome this limitation, we study balanced $(3,3)$-bicliques, the smallest dense extension beyond butterflies in signed bipartite graphs. To the best of our knowledge, this is the first work on counting balanced $(3,3)$-bicliques in signed bipartite graphs, where balance is defined by the absence of unbalanced butterflies. As a baseline, we first adapt and extend the state-of-the-art BCList++ algorithm for unsigned bipartite graphs to incorporate edge signs, which we call \textbf{SBCList++}. We then propose two efficient algorithms: a wedge-centric approach that enforces balance constraints during enumeration and a triplet-centric approach that directly enumerates feasible vertex triplets. Extensive experiments on large real-world datasets demonstrate that the triplet-centric algorithm significantly outperforms the baseline, achieving an average speedup of over $110\times$ compared to SBCList++.

% 
%A bipartite graph is the fundamental model for many complicated networks in real-world applications
Two disjoint sets of entities and their relationship can be modelled as a bipartite graph. Real-life examples include drug-target interaction in biological networks, user-item relationships in e-commerce networks, etc. Motif-based analysis is essential for understanding the structure of large-scale networks, and bipartite graphs are no exception. In contrast to unsigned graphs, motif analysis in signed bipartite graphs has received limited attention. The smallest non-trivial motif in a signed bipartite graph is a balanced $(2,2)$-biclique, often called a balanced butterfly, which captures only local patterns and cannot reveal higher-order relationships. Bipartite motifs have been studied in the literature in the context of signed bipartite graphs, such as maximal biclique, bitruss, and so on. None of these works addresses bipartite motifs with fixed-sized vertex sets, which are often relevant in practical situations. In this work, we study the balanced $(p,q)$-biclique counting problem for small values of $p$ and $q$. As a baseline, we first adapt and extend the state-of-the-art BCList++ algorithm for unsigned bipartite graphs to incorporate edge signs, which we call \textbf{SBCList++}. We then propose two efficient algorithms: \textbf{BBWC}, a wedge-centric approach that enforces balance constraints during enumeration, and \textbf{BBVP}, a vertex-based pruning approach that directly enumerates feasible vertex sets. Extensive experiments on large real-world datasets demonstrate that the vertex-based pruning algorithm, BBVP, significantly outperforms the baseline, achieving an average \textbf{speedup of $561\times$} over SBCList++.

% Signed bipartite graphs are fundamental models in many applications where relationships carry polarity, such as trust, biological, and recommendation networks. Motif-based analysis is essential for understanding the structure, yet in signed bipartite graphs, existing work mainly focuses on local patterns such as balanced butterflies or on 

\end{abstract}

\vspace{10.0cm}

\section{Introduction}
\label{sec:intro}

Bipartite graphs are commonly used in e-commerce user-product interactions, drug-target linkages, and author-paper networks because they naturally depict relationships between two distinct entity types. Formally, a bipartite graph $G(U, V, E)$ has two vertex sets, $U$ and $V$, with edges $e = (u,v)\in E$ connecting $u\in U$ to $v\in V$ (as illustrated in Fig.~\ref{fig:gray_vs_signed}(a)). Significant research interest has been generated by the mining of structural patterns in bipartite networks. Notable patterns include bicliques~\cite{yao2022identifying,chen2022efficient,lyu2020maximum}, bi-truss~\cite{zou2016bitruss,chen2021higher}, and bi-core~\cite{luo2023efficient}, among others. Among these cohesive structures, the \textit{butterfly}~\cite{aksoy2017measuring} (also known as a 4-cycle), i.e., a $2\times 2$ complete bipartite subgraph, is a fundamental motif in bipartite networks and serves as a basic building block of higher-order bipartite structures.

However, many real-world interactions are polarised, unlike classic bipartite graph models. For example, in a drug-target interaction network, a drug can activate or inhibit a biological target, causing positive or negative consequences. Signed interactions reflect the drug's method of action beyond connectivity and cannot be described by an unsigned bipartite graph. Recent works have examined maximal clique enumeration in signed unipartite graphs~\cite{sun2020discovering,chen2020efficient}. To improve the characterisation of cliques in signed unipartite graphs, current methodologies have utilised balancing theory~\cite{heider1946attitudes}. This has led to significant research on balanced triangles, which examines the coherence of positive and negative interactions. In this concept, a triangle is balanced if it has an even number of negative edges, which aligns with common social intuitions.

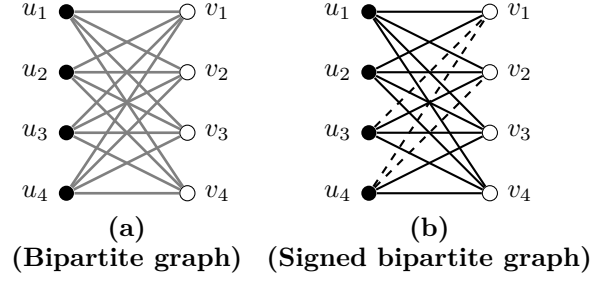
\begin{figure}[t]
    \centering
    \begin{tikzpicture}[scale=0.8]
        % Define node style
        \tikzstyle{blackdot} = [circle, fill=black, minimum size=0.2cm, inner sep=0pt]
        \tikzset{
            whitedot/.style={circle,   draw=black,
            fill=white, minimum size=0.2cm, inner sep=0pt},
        }
        % ----- Left biclique: All solid gray edges -----
        % Left side nodes (U)
        \node[blackdot, label=left:$u_1$] (u0) at (0,1.4) {};
        \node[blackdot, label=left:$u_2$] (u1) at (0,0.4) {};
        \node[blackdot, label=left:$u_3$] (u2) at (0,-0.6) {};
        \node[blackdot, label=left:$u_4$] (u3) at (0,-1.6) {};

        % Right side nodes (V)
        \node[whitedot, label=right:$v_1$] (v0) at (2,1.4) {};
        \node[whitedot, label=right:$v_2$] (v1) at (2,0.4) {};
        \node[whitedot, label=right:$v_3$] (v2) at (2,-0.6) {};
        \node[whitedot, label=right:$v_4$] (v3) at (2,-1.6) {};

        % Draw gray solid edges for full biclique
        \foreach \i in {0,1,2,3} {
            \foreach \j in {0,1,2,3} {
                \draw[line width=1pt, color=gray] (u\i) -- (v\j);
            }
        }

        % ----- Right bipartite graph: Signed structure -----
        % Left nodes (U)
        \node[blackdot, label=left:$u_1$] (bu0) at (5,1.4) {};
        \node[blackdot, label=left:$u_2$] (bu1) at (5,0.4) {};
        \node[blackdot, label=left:$u_3$] (bu2) at (5,-0.6) {};
        \node[blackdot, label=left:$u_4$] (bu3) at (5,-1.6) {};

        % Right nodes (V)
        \node[whitedot, label=right:$v_1$] (bv0) at (7,1.4) {};
        \node[whitedot, label=right:$v_2$] (bv1) at (7,0.4) {};
        \node[whitedot, label=right:$v_3$] (bv2) at (7,-0.6) {};
        \node[whitedot, label=right:$v_4$] (bv3) at (7,-1.6) {};

        % Edges for u1
        \draw[thick] (bu0) -- (bv0);
        \draw[thick] (bu0) -- (bv1);
        \draw[thick] (bu0) -- (bv2);
        \draw[thick] (bu0) -- (bv3);

        % Edges for u2
        \draw[thick] (bu1) -- (bv0);
        \draw[thick] (bu1) -- (bv1);
        \draw[thick] (bu1) -- (bv2);
        \draw[thick] (bu1) -- (bv3);

        % Edges for u3 (with negative)
        \draw[thick, dashed] (bu2) -- (bv0);
        \draw[thick] (bu2) -- (bv1);
        \draw[thick] (bu2) -- (bv2);
        \draw[thick] (bu2) -- (bv3);

        % Edges for u4 (with negative)
        \draw[thick, dashed] (bu3) -- (bv0);
        \draw[thick, dashed] (bu3) -- (bv1);
        \draw[thick] (bu3) -- (bv2);
        \draw[thick] (bu3) -- (bv3);

        % ----- Label (a) and (b) -----
        \node at (1, -2.2) {\textbf{(a)}};
        \node at (1, -2.7) {\textbf{(Bipartite graph)}};
        \node at (6, -2.2) {\textbf{(b)}};
         \node at (6, -2.7) {\textbf{(Signed bipartite graph)}};
    \end{tikzpicture}
    \caption{Example of bipartite graph (grey lines denote edges) and signed bipartite graph (solid/dashed edges denote positive/negative edges, respectively.)}
    \label{fig:gray_vs_signed}
\end{figure}

Existing clique techniques primarily focus on signed unipartite graphs and cannot be directly applied to signed bipartite graphs, since bipartite graphs contain no triangles and consist of two disjoint vertex sets. A signed bipartite graph is defined as $G=(U, V, E)$, where $U$ and $V$ are two partitions and each edge $e \in E$ carries either a positive (``+'') or negative (``-'') sign, as illustrated in Fig.~\ref{fig:gray_vs_signed}(b). In this context, Derr \emph{et al.}~\cite{derr2019balance} introduced signed butterflies, i.e., signed $(2,2)$-bicliques, and defined a butterfly as balanced if it contains an even number of negative edges. Examples of balanced and unbalanced butterflies are shown in Fig.~\ref{figure:bal_unbal}. Subsequent studies have used balanced butterfly-based analysis for tasks such as sign prediction and structural characterization in signed bipartite graphs~\cite{derr2019balance,chung2023maximum}. Despite their usefulness, butterflies capture only limited local structure and cannot describe higher-order structural relationships among more than two nodes within the same partition. In many graph-based tasks, biclique size is not restricted to ($2,2$), motivating the study of larger balanced bicliques. To capture higher-order structure in bipartite graphs, prior studies have explored alternative motifs such as bi-triangles and induced $6$-cycles~\cite{opsahl2013triadic,yang2021efficient,niu2025fast}. However, these motifs are cycle-based rather than clique-based and do not form complete bipartite subgraphs. Therefore, bi-triangles and induced $6$-cycles are fundamentally different from balanced bicliques.

\begin{figure}[ht]
    \centering
    \begin{tikzpicture}[line width=0.2mm][scale=0.5]
        \tikzset{
            blackdot/.style={circle, fill=black, minimum size=0.2cm, inner sep=0pt},
        }
        \tikzset{
            whitedot/.style={circle,   draw=black,
            fill=white, minimum size=0.2cm, inner sep=0pt},
        }

         \draw[dotted, thick] (-0.3, 1.3) rectangle (8.5, -1.4);

        % Butterfly 1 (a)
        \node[blackdot, label={[yshift=1pt]below:$u_2$}] (u0) at (0, 0) {};
        \node[blackdot, label={[yshift=-2pt]above:$u_1$}] (u1) at (0, 0.7) {};
        \node[whitedot, label={[yshift=1pt, xshift=2pt]below:$v_2$}] (v0) at (0.7, 0) {};
        \node[whitedot, label={[yshift=-2pt, xshift=2pt]above:$v_1$}] (v1) at (0.7, 0.7) {};
        \draw[thick] (u0) -- (v0);
        \draw[thick] (u0) -- (v1);
        \draw[thick] (u1) -- (v0);
        \draw[thick] (u1) -- (v1);
        \node at (0.35, -0.75) {(a)};

        % Butterfly 2 (b)
        \node[blackdot, label={[yshift=1pt]below:$u_2$}] (u2) at (1.2, 0) {};
        \node[blackdot, label={[yshift=-2pt]above:$u_1$}] (u3) at (1.2, 0.7) {};
        \node[whitedot, label={[yshift=1pt, xshift=2pt]below:$v_2$}] (v2) at (1.9, 0) {};
        \node[whitedot, label={[yshift=-2pt, xshift=2pt]above:$v_1$}] (v3) at (1.9, 0.7) {};
        \draw[thick] (u2) -- (v2);
        \draw[dashed, thick] (u2) -- (v3);
        \draw[dashed, thick] (u3) -- (v2);
        \draw[thick] (u3) -- (v3);
        \node at (1.55, -0.75) {(b)};

        % Butterfly 3 (c)
        \node[blackdot, label={[yshift=1pt]below:$u_2$}] (u4) at (2.4, 0) {};
        \node[blackdot, label={[yshift=-2pt]above:$u_1$}] (u5) at (2.4, 0.7) {};
        \node[whitedot, label={[yshift=1pt, xshift=2pt]below:$v_2$}] (v4) at (3.1, 0) {};
        \node[whitedot, label={[yshift=-2pt, xshift=2pt]above:$v_1$}] (v5) at (3.1, 0.7) {};
        \draw[dashed] (u4) -- (v4);
        \draw[dashed] (u4) -- (v5);
        \draw[thick] (u5) -- (v4);
        \draw[thick] (u5) -- (v5);
        \node at (2.75, -0.75) {(c)};

        % Butterfly 4 (d)
        \node[blackdot, label={[yshift=1pt]below:$u_2$}] (u6) at (3.6, 0) {};
        \node[blackdot, label={[yshift=-2pt]above:$u_1$}] (u7) at (3.6, 0.7) {};
        \node[whitedot, label={[yshift=1pt, xshift=2pt]below:$v_2$}] (v6) at (4.3, 0) {};
        \node[whitedot, label={[yshift=-2pt, xshift=2pt]above:$v_1$}] (v7) at (4.3, 0.7) {};
        \draw[dashed] (u6) -- (v6);
        \draw[thick] (u6) -- (v7);
        \draw[dashed, thick] (u7) -- (v6);
        \draw[thick] (u7) -- (v7);
        \node at (3.95, -0.75) {(d)};

        % Butterfly 5 (e)
        \node[blackdot, label={[yshift=1pt]below:$u_2$}] (u8) at (4.8, 0) {};
        \node[blackdot, label={[yshift=-2pt]above:$u_1$}] (u9) at (4.8, 0.7) {};
        \node[whitedot, label={[yshift=1pt, xshift=2pt]below:$v_2$}] (v8) at (5.5, 0) {};
        \node[whitedot, label={[yshift=-2pt, xshift=2pt]above:$v_1$}] (v9) at (5.5, 0.7) {};
        \draw[dashed, thick] (u8) -- (v8);
        \draw[dashed] (u8) -- (v9);
        \draw[dashed, thick] (u9) -- (v8);
        \draw[dashed, thick] (u9) -- (v9);
        \node at (5.15, -0.75) {(e)};

        % Extra spacing between (e) and (f)
        % Butterfly 6 (f)
        \node[blackdot, label={[yshift=1pt]below:$u_2$}] (u10) at (6.2, 0) {};
        \node[blackdot, label={[yshift=-2pt]above:$u_1$}] (u11) at (6.2, 0.7) {};
        \node[whitedot, label={[yshift=2pt, xshift=1pt]below:$v_2$}] (v10) at (6.9, 0) {};
        \node[whitedot, label={[yshift=-2pt, xshift=2pt]above:$v_1$}] (v11) at (6.9, 0.7) {};
        \draw[dashed, thick] (u10) -- (v10);
        \draw[thick] (u10) -- (v11);
        \draw[thick] (u11) -- (v10);
        \draw[thick] (u11) -- (v11);
        \node at (6.55, -0.75) {(f)};

        % Butterfly 7 (g)
        \node[blackdot, label={[yshift=1pt]below:$u_2$}] (u12) at (7.4, 0) {};
        \node[blackdot, label={[yshift=-2pt]above:$u_1$}] (u13) at (7.4, 0.7) {};
        \node[whitedot, label={[yshift=1pt, xshift=2pt]below:$v_2$}] (v12) at (8.1, 0) {};
        \node[whitedot, label={[yshift=-2pt, xshift=2pt]above:$v_1$}] (v13) at (8.1, 0.7) {};
        \draw[dashed] (u12) -- (v12);
        \draw[dashed] (u12) -- (v13);
        \draw[dashed] (u13) -- (v12);
        \draw[thick] (u13) -- (v13);
        \node at (7.75, -0.75) {(g)};

        \node at (2.7, -1.15) {\textbf{Balanced}};
        \node at (7.0, -1.15) {\textbf{Unbalanced}};

\draw[dotted, thick] (5.8, 1.3) -- (5.8, -1.4);

    \end{tikzpicture}
    \caption{Illustration of balanced and unbalanced butterflies (solid/dashed edges denote positive/negative edges).}
    \label{figure:bal_unbal}
\end{figure}
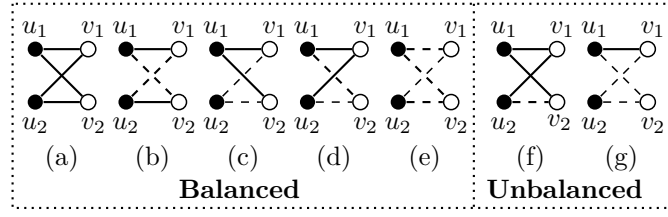

In signed bipartite networks, an important cohesive structure is the \emph{balanced signed biclique}, which extends the concept of balanced butterflies to larger and more complex patterns. Several related structures have been explored, including (maximal) signed bicliques~\cite{sun2023efficient,wang2024efficient,sun2022maximal} and signed bitrusses~\cite{chung2023maximum}, which aim to identify dense or balanced regions for different analytical goals. For example, Sun et al.~\cite{sun2022maximal} developed algorithms for finding maximal balanced bicliques under size constraints $(\tau_u, \tau_v)$, ensuring that each biclique has at least $\tau_u$ vertices on one side and $\tau_v$ vertices on the other side. However, these algorithms are not directly applicable to counting fixed-size balanced bicliques, as we need to discard a large number of bicliques in the post-processing step, which is computationally inefficient. To address these limitations, we study the problem of counting balanced bicliques of size $(p,q)$, where $p$ and $q$ are the sizes of the two partitions. Recently, Kiran et al.~\cite{kiran2024efficient} proposed an algorithm for counting balanced $(2,k)$-bicliques. In the $(2,k)$ case, all butterflies share the same vertex pair, so sign dependencies remain localized. In contrast, a $(p,q)$-biclique contains butterflies formed from many vertex combinations on both sides, and each edge can participate in multiple butterflies. As $p$ and $q$ increase, the number of butterflies grows combinatorially, and the number of interacting sign patterns increases rapidly. Consequently, maintaining the balanced biclique property becomes significantly more challenging, since every butterfly in the structure must satisfy the balance condition simultaneously, requiring global sign consistency across the entire biclique. Beyond structural analysis, balanced $(p,q)$-bicliques have strong practical significance. For instance, in signed drug-target interaction networks, drugs may activate or inhibit biological targets~\cite{torres2016drug}. A fixed balanced $(p,q)$-biclique represents a group of $p$ drugs that exhibit consistent regulatory behavior across $q$ shared targets, revealing coherent multi-drug functional modules of controlled, interpretable size. Such fixed-sized patterns allow systematic comparison of interaction modules across datasets and help identify drug groups with similar mechanisms of action. These structures are particularly valuable in the design of combination therapies.

\subsection{Challenges and Contributions}

The problem of counting balanced \((p, q)\)-bicliques in a signed bipartite graph is computationally challenging, as it requires handling edge signs and incorporating balance theory to ensure that the \((p, q)\)-biclique is stable and free from unbalanced butterflies. Given a bipartite graph $G=(U, V, E)$, a straightforward approach shows that there can be at most ${m\choose p}{n\choose q}$ distinct $(p,q)$-bicliques. Thus, exhaustive counting has a worst-case complexity of approximately $O(m^pn^q)$, which is significantly higher than that of butterfly counting~\cite{wang2019vertex}, where $m = |U|$ and $n = |V|$. Sun et al.~\cite{sun2022maximal} proposed an algorithm for finding balanced bicliques in signed bipartite graphs under size constraints $(\tau_u, \tau_v)$. However, to count balanced bicliques of exact size $(\tau_u, \tau_v)$, one must generate many larger bicliques and filter them, leading to inefficiency. Yang et al.~\cite{yang2023p} designed the BCList++ algorithm to count all $(p,q)$-bicliques, but their method does not consider edge signs. To find balanced bicliques, one would need to enumerate all fixed-size bicliques with BCList++ and then filter out the unbalanced ones. In both approaches, the need to filter a potentially large number of irrelevant bicliques makes them impractical for efficient fixed-size balanced biclique counting. Motivated by these difficulties, we develop algorithms that utilise structural properties of signed bipartite graphs to enable efficient and exact counting. We summarize our main contributions as follows:

\begin{itemize}
    \item We formulate and study the problem of counting balanced $(p,q)$-bicliques in signed bipartite graphs, capturing global balance patterns.
    \item We extend the BCList++ algorithm to signed bipartite graphs, resulting in \textbf{SBCList++} for efficient balanced $(p,q)$-biclique counting.
    \item We propose a wedge-centric algorithm, \textbf{BBWC}, which groups signed wedge patterns by edge-sign configurations and counts only balanced $(p,q)$-bicliques.
    % \item We design a vertex-based pruning algorithm, \textbf{BBVP}, that anchors enumeration on a fixed vertex and reduces redundant balance checks via combinatorial aggregation.
    \item We design a vertex-based pruning algorithm, \textbf{BBVP}, that anchors enumeration at a fixed vertex, prunes infeasible candidates early, and avoids redundant wedge enumerations.
    
    \item We empirically evaluate our algorithms and demonstrate the efficiency and scalability of \textbf{BBWC} and \textbf{BBVP} over baseline algorithms~\footnote{\textit{Source code will be released upon acceptance of the paper}}.
    % \item We develop multi-core versions of both algorithms, \textbf{M-BB3C} and \textbf{MT-BB3C}, and demonstrate their efficiency and scalability on real-world signed bipartite graphs.

        %Note:
%\textit{The implementations of all proposed algorithms will be released as open-source\footnote{\textit{Upon acceptance of the paper}}.}
\end{itemize}

\textbf{Organization of the Paper:}
The remainder of the paper is organized as follows. 
Section~\ref{Sec:rel} reviews related work on signed bipartite graphs. 
Section~\ref{Sec:prob_stmt} presents preliminaries and formally defines the problem. 
Section~\ref{Sec:alg} describes the proposed algorithms for balanced $(p,q)$-biclique counting. 
Section~\ref{Sec:EE} reports the experimental evaluation. 
Section~\ref{Sec:Con} concludes the paper and outlines future research directions.
\section{Related Work}
\label{Sec:rel}
Structural pattern mining in graphs has been extensively studied in large-scale networks, with cliques and bicliques being two fundamental cohesive subgraph structures. While maximal clique enumeration in unipartite graphs has a long history of research, in this literature, we focus on bicliques, as they are closest to the theme of our study.

% we elaborate clique as this is closest to the theme of our study.
In bipartite graphs, a fundamental line of work focuses on butterfly ($(2,2)$-biclique) counting, which serves as a basic building block for higher-order analysis. Numerous efficient algorithms have been proposed for counting butterflies in large bipartite graphs under different computational settings, including sequential, parallel, distributed, and GPU-based frameworks~\cite{wang2014rectangle,sanei2018butterfly,shi2020parallel,xia2024gpu}.
However, these methods are primarily designed for unsigned graphs and capture only local structural patterns. General $(p,q)$-biclique enumeration has also been studied in unsigned bipartite graphs~\cite{yang2023p}, as well as in the context of densest subgraph discovery~\cite {mitzenmacher2015scalable}. Since these methods do not account for edge signs and are enumeration-based, they are not well-suited to efficiently count fixed-size balanced bicliques in signed bipartite graphs. Beyond bicliques, alternative bipartite motifs such as bi-triangles and induced $6$-cycles have been studied to model higher-order connectivity and clustering~\cite{opsahl2013triadic,yang2021efficient,niu2025fast}. These motifs are cycle-based structures and differ fundamentally from bicliques, as they do not form complete bipartite subgraphs and are not designed to capture balance consistency in signed networks.

Several studies have investigated cohesive substructures in signed bipartite graphs. Derr~\emph{et al.}~\cite{derr2019balance} introduced the concept of balanced butterflies and demonstrated their effectiveness in analyzing balance properties in signed bipartite networks. While useful, these approaches remain limited to butterfly-level structures and do not capture higher-order balance consistency. Existing work includes enumeration of maximal signed bicliques~\cite{sun2023efficient,wang2024efficient} and signed bitrusses~\cite{chung2023maximum}, which aim to identify balanced regions under size or density constraints. Recent work has studied balanced $(2,k)$-bicliques in signed bipartite graphs by fixing one side to size two and enforcing balance over butterfly-centred structures~\cite{kiran2024efficient}. While effective for capturing local balance patterns, these methods rely on fixed vertex pairs. To the best of our knowledge, efficient algorithms for counting \emph{balanced} $(p,q)$-bicliques in signed bipartite graphs have not been studied.

% In contrast, our work focuses on counting balanced $(3,3)$-bicliques, which extend balanced butterflies to the smallest dense structure capable of capturing global balance consistency in signed bipartite graphs.

\section{Preliminaries and Problem Definition} 
\label{Sec:prob_stmt}
We consider an undirected signed bipartite graph $G = (U, V, E)$, where $U$ and $V$ are the bipartitions and $E \subseteq U \times V$ represents the edge set. Each edge $(u,v)\in E$ is associated with a sign (label), either ``$+$'' or ``$-$''. The edge with a ``$+$'' (``$-$'') sign denotes the positive (negative) edge, e.g., like (dislike), trust (distrust), etc. For a vertex $u\in U$, let $deg(u)$ denote the degree of $u$ and $\Gamma(u)$ denote the neighbors of $u$ in $G$. Now, we define some basic notions for our problem.

\begin{definition}[\textbf{Butterfly}~\cite{sun2022maximal}]
    Given a bipartite graph $G=(U,V,E)$, a butterfly is a cycle of length $4$ induced by $4$ nodes $(u_i, u_j, v_i, v_j)$, where $u_i, u_j\in U$ and $v_i, v_j\in V$, such that $u_i$ and $u_j$ are all connected to $v_i$ and $v_j$ in $G$. A \textbf{signed butterfly} in a signed bipartite graph is a butterfly where all the edges are either positive or negative.
\label{def:butterfly}
\end{definition}

\begin{definition}[\textbf{Balanced biclique}]
Given a signed bipartite graph $G = (U,V,E)$, a balanced biclique is an induced subgraph $B$ of $G$ that satisfies:
\begin{itemize}
    \item \textbf{Cohesiveness constraint:}~$B$ is a signed biclique of $G$;
    \item \textbf{Balance constraint:}~$B$ does not contain any unbalanced butterfly.
\end{itemize}
\label{def: bsb}
\end{definition}

\begin{definition}[\textbf{$(p,q)$-Biclique}~\cite{yang2023p}]~Given a bipartite graph $G=(U,V,E)$, and two integer parameters $p$ and $q$, a $(p, q)$-biclique $B(L,R)$ is a biclique of $G$ with $|L|=p$ and $|R|=q$.
\end{definition}

\begin{definition}[\textbf{Balanced $(p,q)$-Biclique}]~A $(p,q)$-biclique in a signed bipartite graph is balanced if there are no unbalanced butterflies in it.
\label{bal-2-3-bcl}
\end{definition}

Now, we are ready to define the problem statement as follows.

\textbf{Problem Statement:} 
Given a signed bipartite graph $G=(U, V, E)$ and two positive integers $p$ and $q$, we study the problem of counting balanced $(p,q)$-bicliques present in $G$, where $p$ vertices are selected from one partition and $q$ from the other. When \( p = 1 \) or \( q = 1 \), the resulting structures degenerate into star-shaped patterns that do not reflect cohesive biclique structures. Therefore, we focus on the case where \( p > 1 \) and \( q > 1 \).

\section{Proposed Methodologies}
\label{Sec:alg}
In this section, we describe our algorithms for counting balanced $(p,q)$-bicliques in signed bipartite graphs.

\subsection{Baseline}
\label{sec:baseline}

A straightforward approach to counting balanced $(p,q)$-bicliques is to first enumerate all \pqbc and then filter out those that do not satisfy the balance criterion. A baseline for this purpose is the state-of-the-art algorithm BCList++~\cite{yang2023p}, which was originally designed for counting and enumerating \pqbc in unsigned bipartite graphs. We extend BCList++ to support signed bipartite graphs and refer to the enhanced version as \textbf{SBCList++}. Although this approach correctly counts all balanced bicliques, it incurs substantial redundant computation by generating many unbalanced bicliques that are later discarded. For example, if the signed bipartite graph contains no balanced bicliques, the algorithm still exhaustively enumerates candidate bicliques before filtering them, leading to significant overhead. Therefore, the baseline approach is inefficient. To address this drawback, we design more efficient algorithms that incorporate edge-sign information directly into the counting process, thereby avoiding the need to enumerate unbalanced structures.

\subsection{Wedge-Centric Algorithm}
\label{sec:wedge}
In this section, we design an efficient algorithm for counting balanced \((p, q)\)-bicliques. The high-level idea of our algorithm is to group the \spwedge patterns, respecting the sign of the edges, into appropriate buckets and systematically combine the \spwedge patterns to count only balanced $(p, q)$-bicliques.

\begin{definition}[\textbf{Vertex priority}~\cite{wang2019vertex}]
For any two vertices $a,b \in U \cup V$, the vertex priority $\rho(\cdot)$ is
defined such that $\rho(a)>\rho(b)$ if $deg(a)>deg(b)$, or if
$deg(a)=deg(b)$ and $id(a)>id(b)$, where $id(\cdot)$ denotes the vertex ID.
\label{def-priority}
\end{definition}

\begin{definition}[\textbf{$p$-star}]
Let $G=(U, V, E)$ be a bipartite graph. A \textbf{$p$-star} is a subgraph induced by
$p+1$ vertices consisting of a center vertex adjacent to $p$ distinct vertices,
where the center has degree $p$, and each of the remaining vertices has degree
$1$ in the subgraph.
\end{definition}

\begin{definition}[\textbf{$p$-wedge}]
\label{def:p-wedge}
Given a bipartite graph $G=(U,V,E)$, a \textbf{$p$-wedge} is a $(p+1)$-tuple
$\langle u_1, w_1, w_2, \ldots, w_{p-1}, v_1\rangle$ such that
$u_1, w_1, w_2, \ldots, w_{p-1}$ belong to the same partition,
$v_1$ belongs to the opposite partition, and the vertices
$\{u_1, w_1, \ldots, w_{p-1}\}$ together with center vertex $v_1$
induce a $p$-star with $v_1$ as the center. In a signed bipartite graph, a \textbf{signed $p$-wedge} is a $p$-wedge in which
each of the $p$ incident edges $(v_1,u_1)$ and $(v_1,w_i)$ carries a sign
(positive or negative). An example of a signed $p$-wedge is shown in
Figure~\ref{fig:wedge_structures} (a).
\end{definition}

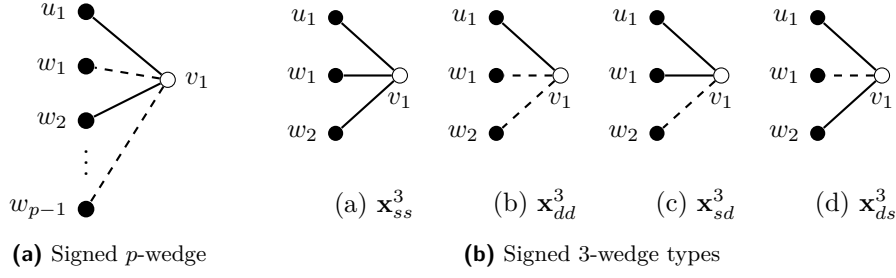
\begin{figure}[t]
\centering

\begin{subfigure}{0.30\linewidth}
\centering
\begin{tikzpicture}[scale=0.9]
\tikzstyle{blackdot} = [circle, fill=black, minimum size=0.22cm, inner sep=0pt]
\tikzset{whitedot/.style={circle, draw=black, fill=white, minimum size=0.2cm, inner sep=0pt}}

\node[whitedot, label=right:$v_1$] (v)  at (2.2, 0) {};
\node[blackdot, label=left:$u_1$] (u)  at (1.0, 1.0) {};
\node[blackdot, label=left:$w_1$] (w1) at (1.0, 0.2) {};
\node[blackdot, label=left:$w_2$] (w2) at (1.0, -0.6) {};
\node at (1.0, -1.1) {$\vdots$};
\node[blackdot, label=left:$w_{p-1}$] (wp) at (1.0, -1.9) {};

\draw[line width=0.8pt] (v) -- (u);
\draw[line width=0.8pt,dashed] (v) -- (w1);
\draw[line width=0.8pt] (v) -- (w2);
\draw[line width=0.8pt,dashed] (v) -- (wp);
\end{tikzpicture}
\caption{Signed $p$-wedge}
\end{subfigure}
\hspace{-0.5cm}
\begin{subfigure}{0.65\linewidth}
\centering
\begin{tikzpicture}[scale=0.85]
\tikzstyle{blackdot} = [circle, fill=black, minimum size=0.2cm, inner sep=0pt]
\tikzset{whitedot/.style={circle, draw=black, fill=white, minimum size=0.2cm, inner sep=0pt}}

% ---------- (a) ss ----------
\node[whitedot, label=below:$v_1$] (v1a) at (-2, 0) {};
\node[blackdot, label=left:$u_1$] (u1a) at (-3, 0.9) {};
\node[blackdot, label=left:$w_1$] (w1a) at (-3, 0) {};
\node[blackdot, label=left:$w_2$] (w2a) at (-3, -0.9) {};
\draw[line width=0.8pt] (v1a) -- (u1a);
\draw[line width=0.8pt] (v1a) -- (w1a);
\draw[line width=0.8pt] (v1a) -- (w2a);
\node at (-2.4, -2.0) {(a) $\mathbf{x}^3_{ss}$};

% ---------- (b) dd ----------
\node[whitedot, label=below:$v_1$] (v1b) at (0.5, 0) {};
\node[blackdot, label=left:$u_1$] (u1b) at (-0.5, 0.9) {};
\node[blackdot, label=left:$w_1$] (w1b) at (-0.5, 0) {};
\node[blackdot, label=left:$w_2$] (w2b) at (-0.5, -0.9) {};
\draw[line width=0.8pt] (v1b) -- (u1b);
\draw[line width=0.8pt,dashed] (v1b) -- (w1b);
\draw[line width=0.8pt,dashed] (v1b) -- (w2b);
\node at (0.1, -2.0) {(b) $\mathbf{x}^3_{dd}$};

% ---------- (c) sd ----------
\node[whitedot, label=below:$v_1$] (v1c) at (3, 0) {};
\node[blackdot, label=left:$u_1$] (u1c) at (2, 0.9) {};
\node[blackdot, label=left:$w_1$] (w1c) at (2, 0) {};
\node[blackdot, label=left:$w_2$] (w2c) at (2, -0.9) {};
\draw[line width=0.8pt] (v1c) -- (u1c);
\draw[line width=0.8pt] (v1c) -- (w1c);
\draw[line width=0.8pt,dashed] (v1c) -- (w2c);
\node at (2.6, -2.0) {(c) $\mathbf{x}^3_{sd}$};

% ---------- (d) ds ----------
\node[whitedot, label=below:$v_1$] (v1d) at (5.5, 0) {};
\node[blackdot, label=left:$u_1$] (u1d) at (4.5, 0.9) {};
\node[blackdot, label=left:$w_1$] (w1d) at (4.5, 0) {};
\node[blackdot, label=left:$w_2$] (w2d) at (4.5, -0.9) {};
\draw[line width=0.8pt] (v1d) -- (u1d);
\draw[line width=0.8pt,dashed] (v1d) -- (w1d);
\draw[line width=0.8pt] (v1d) -- (w2d);
\node at (5.1, -2.0) {(d) $\mathbf{x}^3_{ds}$};

\end{tikzpicture}

\caption{Signed $3$-wedge types}
\end{subfigure}

\caption{An example of \pwedge and signed 3-wedge with different configurations.}
\label{fig:wedge_structures}
\end{figure}

In a signed bipartite graph, we will have different types of \spwedge{s} respecting the ordering of the vertices and the sign of the edges. For any \spwedge $\omega = \langle u_1, w_1, w_2, ..., w_{p-1}, v_1\rangle$, we assume the ordering $u_1, w_1, w_2, ..., w_{p-1}$ such that $\rho(u_1)>\rho(w_1)>\rho(w_2)>\cdots>\rho(w_{p-1})$. Next, we consider the sign on the edges of a \spwedge and define the type of a \spwedge $\omega$ as $\mathbf{x}^p_{\kappa_j}$ where $\kappa_j = c_1c_2...c_{p-1}$ with each $c_i$ is either $s$ or $d$. For an index $i$, $c_i = s$ if the sign of edge $(u_1, v_1)$ is same as the sign of the edge $(v_1, w_i)$, and $d$ otherwise. For example, Figure.~\ref{fig:wedge_structures}(b) shows different types of \stredge{s}.

\begin{lemma}\label{lemma-1}
    All \spwedge{s} in a balanced $(p, q)$biclique are of the same type.
\end{lemma}

\begin{proof}
    % We will prove this by contradiction. Suppose there is a balanced $(p, q)$-biclique $b$ where there are two different types of \pwedge $\mathbf{x}^p_{\kappa} = \langle u, w_1, w_2, ..., w_{p-1}, v\rangle$ and $\mathbf{x}^p_{\kappa'} = \langle u, w_1, w_2, ..., w_{p-1}, v'\rangle$. Suppose in the representations of $\kappa = c_1c_2...c_{p-1}$ and $\kappa'=c_1'c_2'...c_{p-1}'$, $i$ is the first index where $c_i\neq c_i'$. Without loss of generality, assume that $c_i=s$ and $c_i'=d$. Then, the butterfly $(u, w_i, v, v')$ is an unbalanced butterfly. This is a contradiction.

 The detailed proof is given in Appendix~\ref{app:spwedge-type-proof}.

% \hfill $\square$
\end{proof}

 \textbf{Algorithm:} The algorithm first selects the smaller partition between $U$ and $V$ as the anchor set $S$. For each vertex $u\in S$,
 we enumerate signed $p$-wedges $\omega = \langle u,w_1, w_2, \dots, w_{p-1},v \rangle$, where $w_1, \dots, w_{p-1} \in S$ and $v$ belongs to the opposite direction. According to Lemma~\ref{lemma-1}, each signed $p$-wedge can be classified into one of $2^{p-1}$ types based on the sign pattern of the $p-1$ edges incident to $V$ (each being of $s$ or $d$). We maintain $2^{p-1}$ buckets, where bucket $B_i$ corresponds to wedge $\mathbf{x}^p_{\kappa_i}$. Each bucket stores a mapping from a wedge $\omega$ (defined by its anchor-side vertices $\{u,w_1, \dots, w_{p-1}\}$) to the number of vertices in the opposite partition that are adjacent to all vertices of $\omega$. In other words, $B_i[\omega]$ counts the number of common neighbors that complete $\omega$ into a $(p,q)$-biclique. After processing all wedges rooted at $u$, for each wedge $\omega$ stored in bucket $B-i$, we add $\binom{B_i[\omega]}{q}$ to the total count, since any choice of $q$ common neighbors forms a balanced $(p,q)$-biclique. Summing over all buckets yields the total number of balanced $(p,q)$-bicliques. A detailed pseudocode of this procedure is presented in Algorithm~\ref{algo:bucketpq}.

\begin{algorithm}[t]
\small
\DontPrintSemicolon
	\SetKwInOut{Input}{Input}
	\SetKwInOut{Output}{Output}
	\Input{$G$$(U,V,E)$: signed bipartite graph.}
	\Output{$b_{pq}$: number of balanced  $(p, q)$ bicliques.}
        calculate $\rho(u)$ for each $u\in U$  [\textbf{Definition}~\ref{def-priority}]\\
        $b_{pq}\gets 0$\;
        Select the smaller partition $S\in\{U,V\}$\;     
       \ForEach{$u\in S$}{
         $B_1 , B_2, ..., B_{2^{p-1}} \gets \phi$\;
        \ForEach{$v \in \Gamma(u)$}{
       \ForEach{$(w_1, w_2, ..., w_{p-1}) \in \Gamma(v) \mbox{ such that } \rho(u) > \rho(w_1) >\cdots>\rho(w_{p-1})$}{
                \lIf{$\omega = \langle u, w_1, w_2, ..., w_{p-1}, v\rangle$ is of type $\mathbf{x}^p_{\kappa_i}$}{
                    $B_i[\omega]++$
                }
                
        }
        }
        \ForEach{$i\in [1, ..., 2^{p-1}]$ such that $B_i$ is non-empty}{
            \ForEach{$\omega\in B_i$}{
                $b_{pq} \gets b_{pq} + \binom{B_i[\omega]}{q}$    
            }
        }

        }
 % }

\caption{\textbf{BBWC}: Wedge-Centric Balanced $(p,q)$-biclique Counting}
\label{algo:bucketpq}
\end{algorithm}

\begin{theorem}
    The worst case time complexity of Algorithm~\ref{algo:bucketpq} for counting \bpqbc in a signed bipartite graph $G = (U, V, E)$ with $|U|=m, |V|=n$ is $O(m\Delta\binom{\Delta}{p})$ where $\Delta$ is the maximum degree of $G$.
\end{theorem}

\begin{proof}
% The time complexity has two parts. First, we compute the complexity of counting the number of vertices connected to each \spwedge in the respective bucket. For this, we construct a vertex subset of size $p$ within the $2$-hop neighborhood of a vertex $u$. The worst-case complexity of this counting is $\Delta\times\binom{\Delta}{p}$ as $\Delta$ is the maximum degree in the graph. For choosing a specific \spwedge respecting the vertex order, we can sort the vertices using counting sort in time $O(m)$. In the second phase of the algorithm, we count the number of \bpqbc using the information of the number of common vertices for each \spwedge (Lines $8$-$10$). For this, we require $\Delta\times\binom{\Delta}{p}$ entries across the buckets, since we have processed these many \spwedge{s} in the first phase of the algorithm. Thus, overall time complexity becomes $O(m\Delta\binom{\Delta}{p})$.
A detailed proof is provided in Appendix~\ref{app:complexity-proof}
 
\end{proof}

%\subsection{Triplet-Centric Algorithm}
\subsection{Optimization: vertex-based pruning}
\label{sec:triplet}
% Although the wedge-centric algorithm correctly and systematically counts balanced $(p,q)$-bicliques, its signed $p$-wedge enumeration process can repeatedly generate the same anchor-side vertex groups via distinct neighbors, thereby increasing computational cost. To improve efficiency, we further develop a triplet-centric approach that first identifies candidate vertex sets with sufficient shared neighbors and then directly enumerates higher-order structures from these sets. This avoids repeated intermediate expansions while preserving correctness. A detailed illustration is provided in Appendix~\ref{appndx:wedge-triplet}.

Although the wedge-centric algorithm correctly counts balanced $(p,q)$-bicliques, it is computationally expensive because it enumerates many signed $p$-wedges. Consider the signed bipartite graph in Figure~\ref{fig:wedge-triplet-comp}, where we aim to count balanced $(3,3)$-bicliques anchored at $u_1$. In the wedge-centric approach, each neighbor $v_i \in \Gamma(u_1)$ is processed independently, and all signed $3$-wedges of the form $(u_1, v_i, (u_j, u_k))$ with $u_j, u_k \in \Gamma(v_i)\setminus\{u_1\}$ are enumerated. This leads to six wedge enumerations from $u_1$; three through $v_1$ and one each via $v_2$, $v_3$, and $v_4$ (shown in Figure~\ref{fig:wedge-triplet-comp}(a)). In contrast, the vertex-based pruning approach first constructs the candidate set $C(u_1)=\{u_j \mid |N(u_1) \cap N(u_j)| \geq 3\},$ which yields $C(u_1)=\{u_2,u_3,u_4\}$. It then enumerates vertex triplets $(u_1,u_j,u_k)$ with $u_j,u_k \in C(u_1)$ only once, producing just three candidates: $(u_1,u_2,u_3)$, $(u_1,u_2,u_4)$, and $(u_1,u_3,u_4)$, (as shown in Figure~\ref{fig:wedge-triplet-comp}(b)). The detailed vertex-based pruning algorithm and pseudocode are presented below.

% For example, consider the signed bipartite graph in which we aim to count balanced $(3,3)$-bicliques, shown in Figure~\ref{fig:wedge-triplet-comp}. We first explain how the wedge-centric algorithm processes this graph when anchored at vertex $u_1$. In the wedge-centric algorithm, each neighbor $v_i \in N(u_1)$ is processed independently. For a fixed $v_i$, the algorithm enumerates all signed $3$-wedges of the form $(u_1, v_i, (u_j, u_k))$ with $u_j, u_k \in N(v_i)\setminus\{u_1\}$. In the example, this results in six wedge enumerations from $u_1$: three through $v_1$ and one each through $v_2$, $v_3$, and $v_4$ (shown in Figure~\ref{fig:wedge-triplet-comp}(a)). In contrast, using the same graph shown in Figure~\ref{fig:wedge-triplet-comp} and the same anchor vertex $u_1$, the Vertex-based pruning algorithm first computes the candidate set $C(u_1)=\{u_j \mid |N(u_1)\cap N(u_j)| \geq 3\},$
% which, in this example, yields $C(u_1)=\{u_2,u_3,u_4\}$. Next, instead of expanding through each neighbor $v_i$ independently, the algorithm enumerates triplets of the form $(u_1,u_j,u_k)$ with $u_j,u_k \in C(u_1)$ exactly once. For the running example, only three candidate triplets are generated: $(u_1,u_2,u_3)$, $(u_1,u_2,u_4)$, and $(u_1,u_3,u_4)$,  shown in Figure~\ref{fig:wedge-triplet-comp}(b). The detailed vertex-based pruning algorithm and pseudocode are presented below.

\begin{figure*}[ht]
\centering
\begin{tikzpicture}[scale=0.70, every node/.style={scale=0.70}]

% =====================
% Styles
% =====================
\tikzset{
  u/.style={circle, fill=black, minimum size=2.5mm, inner sep=0pt},
  v/.style={circle, draw=black, fill=white, minimum size=2.5mm, inner sep=0pt},
  edge/.style={gray, thick},
  pairbox/.style={draw, rounded corners=4pt, thick, inner sep=4pt},
  pair23/.style={pairbox, fill=blue!12},
  pair24/.style={pairbox, fill=green!12},
  pair34/.style={pairbox, fill=orange!15},
  trip/.style={draw, thick, rounded corners=4pt, inner sep=5pt, fill=orange!18}
}

% =====================================================
% TOP: INPUT SIGNED BIPARTITE GRAPH
% =====================================================
\begin{scope}[shift={(-3.5,0)}, scale=1.1]

% U nodes
\node[u,label=above:$u_1$] (u1) at (-1.5,1.2) {};
\node[u,label=above:$u_2$] (u2) at (-0.5,1.2) {};
\node[u,label=above:$u_3$] (u3) at (0.5,1.2) {};
\node[u,label=above:$u_4$] (u4) at (1.5,1.2) {};

% V nodes
\node[v,label=below:$v_1$] (v1) at (-1.5,0) {};
\node[v,label=below:$v_2$] (v2) at (-0.5,0) {};
\node[v,label=below:$v_3$] (v3) at (0.5,0) {};
\node[v,label=below:$v_4$] (v4) at (1.5,0) {};

% -----------------
% Explicit edges
% -----------------

% u1 -> v1,v2,v3,v4
\draw[edge] (u1)--(v1);
\draw[edge] (u1)--(v2);
\draw[edge] (u1)--(v3);
\draw[edge] (u1)--(v4);

% u2 -> v1,v2,v3
\draw[edge] (u2)--(v1);
\draw[edge] (u2)--(v2);
\draw[edge] (u2)--(v3);

% u3 -> v1,v2,v4
\draw[edge] (u3)--(v1);
\draw[edge] (u3)--(v2);
\draw[edge] (u3)--(v4);

% u4 -> v1,v3,v4
\draw[edge] (u4)--(v1);
\draw[edge] (u4)--(v3);
\draw[edge] (u4)--(v4);

% -----------------
% Caption
% -----------------
\node at (0,-0.7) {\textbf{Input: Signed Bipartite Graph}};

\end{scope}

% =====================================================
% (a) LEFT-BOTTOM: WEDGE-CENTRIC
% =====================================================
\begin{scope}[shift={(-8,-1.7)}, scale=0.9]

\node[u,label=above:$u_1$] (ua) at (0,0) {};
\node at (-5,0.5) {Start vertex $u_1$:};
\node[v,label=left:$v_1$] (va1) at (-3,-0.8) {};
\node at (-5.6,-0.8) {\textbf{$\Gamma(u_1)$}};

\node[v,label=left:$v_2$] (va2) at (-1,-0.8) {};
\node[v,label=left:$v_3$] (va3) at (1,-0.8) {};
\node[v,label=left:$v_4$] (va4) at (3,-0.8) {};

\draw[edge] (ua)--(va1);
\draw[edge] (ua)--(va2);
\draw[edge] (ua)--(va3);
\draw[edge] (ua)--(va4);

\node[pair23] (a23a) at (-4.2,-2.0) {$(u_2,u_3)$};
\node[align=center] at (-5.8,-2.4) {
  \textbf{$\Gamma(v_i)$} \\
  % $[\rho(u_1) > \rho(u_{i}> \rho(u_{i+1})] $
};
\node[align=center] at (-3.4,-2.9) {
  \textbf{$\rho(u_1) > \rho(u_{i)})> \rho(u_{i+1})$;} 
  \\ $i>1$
};

\node[pair24] (a24a) at (-2.5,-2.0) {$(u_2,u_4)$};
\node[pair34] (a34a) at (-0.8,-2.0) {$(u_3,u_4)$};

\node[pair23] (a23b) at (0.9,-2.0) {$(u_2,u_3)$};
\node[pair24] (a24b) at (2.6,-2.0) {$(u_2,u_4)$};
\node[pair34] (a34b) at (4.3,-2.0) {$(u_3,u_4)$};

\draw[edge] (va1)--(a23a);
\draw[edge] (va1)--(a24a);
\draw[edge] (va1)--(a34a);

\draw[edge] (va2)--(a23b);
\draw[edge] (va3)--(a24b);
\draw[edge] (va4)--(a34b);

\node[
    anchor=west,
    align=left,
    font=\small
] (wedgelist) at (-2.6, -5.0) {
    \emph{Anchor: $u_1$} \\[-1pt]
    
    $v_1:$ $(u_1,v_1,(u_2,u_3))$\\
    \phantom{$v_1:$} $(u_1,v_1,(u_2,u_4))$\\
    \phantom{$v_1:$} $(u_1,v_1,(u_3,u_4))$\\[-1pt]
    $v_2:$ $(u_1,v_2,(u_2,u_3))$\\[-1pt]
    $v_3:$ $(u_1,v_3,(u_2,u_4))$\\[-1pt]
    $v_4:$ $(u_1,v_4,(u_3,u_4))$
};

\node[
    anchor=west,
    align=left,
    font=\small,
    draw,
    dotted,
    thick,
    inner sep=4pt
] (wedgecount) at (-2.6, -7.0) {
    \emph{\textbf{Total 3-wedge enumerated = $6$}}
};

\node[
    anchor=west,
    align=left,
    font=\small
] at (-2.6, -8.5) {
    \emph{\textbf{Counting:}} \\[-1pt]
    $(u_2,u_3) = 2$ \\[-1pt]
    $(u_2,u_4) = 2$ \\[-1pt]
    $(u_3,u_4) = 2$ \\[2pt]
    \emph{\textbf{Balanced $(3,3)$-bicliques = $0$}}
};

\node at (0,-12.0) {\textbf{(a) Wedge-centric}};
\end{scope}

% =====================================================
% (b) RIGHT-BOTTOM: VERTEX-FIRST
% =====================================================
\begin{scope}[shift={(0.5,-1.7)}, scale=0.9]

\node[u,label=above:$u_1$] (ub) at (0,0) {};
\node at (3.2,0.5) {Start vertex $u_1$};

\node[v,label=left:$v_1$] (vb1) at (-3,-0.8) {};
\node at (4,-0.8) {\textbf{$\Gamma(u_1)$}};

\node[v,label=left:$v_2$] (vb2) at (-1,-0.8) {};
\node[v,label=left:$v_3$] (vb3) at (1,-0.8) {};
\node[v,label=left:$v_4$] (vb4) at (3,-0.8) {};

\draw[edge] (ub)--(vb1);
\draw[edge] (ub)--(vb2);
\draw[edge] (ub)--(vb3);
\draw[edge] (ub)--(vb4);

\node[u,label=below:$u_2$] (u12) at (-3.6,-2.0) {};
\node[u,label=below:$u_3$] (u13) at (-3.0,-2.0) {};
\node[u,label=below:$u_4$] (u14) at (-2.4,-2.0) {};
\node[u,label=below:$u_2$] (u22) at (-1.4,-2.0) {};
\node[u,label=below:$u_3$] (u23) at (-0.6,-2.0) {};
\node[u,label=below:$u_2$] (u32) at (0.6,-2.0) {};
\node[u,label=below:$u_4$] (u34) at (1.4,-2.0) {};
\node[u,label=below:$u_3$] (u43) at (2.6,-2.0) {};
\node[u,label=below:$u_4$] (u44) at (3.4,-2.0) {};

\draw[edge] (vb1)--(u12);
\draw[edge] (vb1)--(u13);
\draw[edge] (vb1)--(u14);
\draw[edge] (vb2)--(u22);
\draw[edge] (vb2)--(u23);
\draw[edge] (vb3)--(u32);
\draw[edge] (vb3)--(u34);
\draw[edge] (vb4)--(u43);
\draw[edge] (vb4)--(u44);

\node[
  draw, thick, inner sep=6pt,
  minimum width=7.2cm, minimum height=0.8cm,
  fit=(u12)(u13)(u14)(u22)(u23)(u32)(u34)(u43)(u44),
   yshift=-4pt
] (ubox) {};

\node[align=center] at (4.4,-2.4) {
  \textbf{$\Gamma(v_i)$}
};

\node[align=center] at (2.6,-3.2) {
  $|N(u_1)\cap N(u_i)| \ge 3;$ \\
  $i>1$
};

\node[trip] (T) at (0,-3.8) {$u_2,u_3,u_4$};

\node (X) at (-2.8,-3.8) {Candidate set (C):};

\draw[
  ->, thick, double, double distance=1.2pt, draw=gray!60
] (ubox.south) -- (T.north);

\node[pair23] (r23) at (-1.5,-5.2) {$u_2,u_3$};
\node[pair24] (r24) at (0,-5.2) {$u_2,u_4$};
\node[pair34] (r34) at (1.5,-5.2) {$u_3,u_4$};

\draw[edge] (T)--(r23);
\draw[edge] (T)--(r24);
\draw[edge] (T)--(r34);

\node[
    anchor=west,
    align=left,
    font=\small
] (triplets) at (-2, -7.0) {
 \emph{Anchor: $u_1$} \\[-1pt]
    $(u_1,u_2,u_3)$\\[-1pt]
    $(u_1,u_2,u_4)$\\[-1pt]
    $(u_1,u_3,u_4)$
};
\node[
    anchor=west,
    align=left,
    font=\small,
    draw,
    dotted,
    thick,
    inner sep=4pt
] (tripcount) at (-2, -8.2) {
    \emph{\textbf{Total triplets= $3$}}
};

\node[
    anchor=west,
    align=left,
    font=\small
] at (-2, -10.1) {
\emph{\textbf{Counting:}} \\
 $N(u_1)\cap N(u_i)\cap N(u_j)>= 3$ \\
    $\{u_1,u_2,u_3\}$  $\{v_1,v_2\}$ $=2$ \\[1pt]
    $\{u_1,u_2,u_4\}$  $\{v_1,v_3\}$ $=2$\\[1pt]
    $\{u_1,u_3,u_4\}$  $\{v_1,v_4\}$ $=2$  \\[3pt]

     \emph{\textbf{Balanced $(3,3)$-bicliques = $0$}}
};

\node at (0,-12.0) {\textbf{(b) Vertex-based pruning}};
\end{scope}

\end{tikzpicture}

\caption{An example illustrating wedge-centric and vertex-based pruning algorithms.}
\label{fig:wedge-triplet-comp}
\end{figure*}
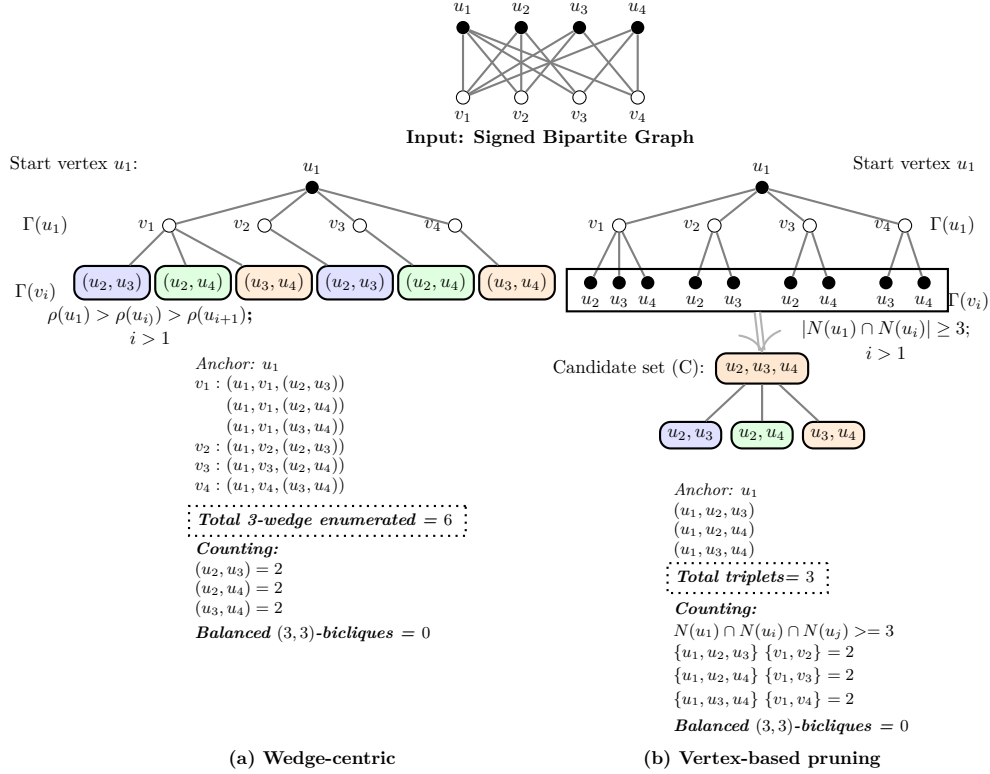

\begin{algorithm}[ht]
\DontPrintSemicolon
\caption{\textbf{BBVP}: Vertex-based Pruning Balanced $(p,q)$-Biclique Counting}
\label{alg:bbvp}
\KwIn{Signed bipartite graph $G=(U,V,E)$}
\KwOut{$b_{pq}$: number of balanced $(p,q)$-bicliques}

Select the smaller partition $S \in \{U,V\}$; $b_{pq} \gets 0$\;

\ForEach{$u \in S$}{

    % Initialize empty map $\mathsf{cnt}$\;
    % Initialize empty map $\mathsf{L}$\;
      Initialize two empty maps $\mathsf{cnt}$ , $\mathsf{L}$\;

    \ForEach{$v \in \Gamma(u)$}{
        \ForEach{$w \in \Gamma(v)$ such that $w \in S$ and $\rho(w) > \rho(u)$}{
            $\mathsf{cnt}[w] \gets \mathsf{cnt}[w] + 1$ \tcp*{\textcolor{blue}{shared neighbors with $u$}}
        }
    }

    Let $C = \{w \in S \mid \mathsf{cnt}[w] \ge q\}$ \tcp*{\textcolor{blue}{necessary pairwise condition}}
    \If{$|C| < p-1$}{continue}

    \ForEach{$w \in C$}{
        $\mathsf{L}[w] \gets \Gamma(u) \cap \Gamma(w)$  \tcp*{\textcolor{blue}{common neighbors with $u$}} %Sort $\mathsf{L}[w]$\;
    }

    \ForEach{$A = (w_1,\dots,w_{p-1}) \subset C$ with $w_1 < \dots < w_{p-1}$}{

        Let $T = \bigcap_{w \in A} \mathsf{L}[w]$ \tcp*{\textcolor{blue}{adjacent to all anchors}}
        \If{$|T| < q$}{continue}

        Initialize empty map $\mathbf{x}$ \tcp*{\textcolor{blue}{signed $p$-wedge type count}}

        \ForEach{$v \in T$}{

            Compute $\kappa = c_1c_2...c_{p-1}$ where $c_i = s$ when $\text{sign}(u, v) = \text{sign}(v, w_i)$ and $c_i = d$ otherwise\;
            %$\sigma(v) = (\text{sign}(u,v), \text{sign}(w_1,v), \dots, \text{sign}(w_{p-1},v))$\;
            %Canonicalize $\sigma(v)$ under global sign flip to obtain type $\tau$\;
            $\mathbf{x}[\kappa] \gets \mathbf{x}[\kappa] + 1$\;
        }

        \ForEach{$p$-wedge type $\kappa$ in $\mathbf{x}$}{
            $b_{pq} \gets b_{pq} + \binom{\mathbf{x}[\kappa]}{q}$\;
        }
    }
}

\Return $b_{pq}$\;
\end{algorithm}

\textbf{Algorithm:} 
The vertex-based pruning algorithm, presented in Algorithm~\ref{alg:bbvp}, first selects the smaller partition between $U$ and $V$, denoted $S$, as the anchor side to reduce the search space. A global counter $b_{pq}$ is initialized to zero. The algorithm then iterates over each vertex $u \in S$, treating $u$ as an anchor vertex. For every anchor, two local data structures are maintained: a counter map $\mathsf{cnt}$, which records the number of shared neighbors between $u$ and other vertices in $S$, and a map $\mathsf{L}$, which stores common-neighbor lists required during biclique construction. For each neighbor $v \in \Gamma(u)$, the algorithm explores vertices $w \in \Gamma(v)$ that also belong to $S$ and satisfy the ordering constraint $\rho(w) > \rho(u)$, ensuring that each biclique is counted exactly once. The counter $\mathsf{cnt}[w]$ is incremented for each such occurrence, so that $\mathsf{cnt}[w]=|\Gamma(u)\cap\Gamma(w)|$, thereby recording the number of neighbors shared by $u$ and $w$. After this step, a candidate set $C = \{\,w \in S \mid \mathsf{cnt}[w] \ge q \,\}$ is formed, containing vertices that satisfy the necessary pairwise condition for forming a $(p,q)$-biclique with anchor $u$. If $|C| < p-1$, then $u$ cannot participate in any such biclique, and the algorithm proceeds to the next anchor, discarding the vertex $u$. Otherwise, for each $w \in C$, the algorithm maintains a list $\mathsf{L}[w]$, the set of vertices directly connected to both $u$ and $w$. 

The algorithm then enumerates all ordered $(p-1)$-subsets $A=(w_1, \dots, w_{p-1}) \subset C$. Together with $u$, these vertices define an $S$ side partition with $p$ vertices of a candidate biclique. Their common neighbors $T = \bigcap_{w \in A} \mathsf{L}[w]$ are computed, which represents the set of vertices in the other partition that are adjacent to all $p$ anchor-side vertices. If $|T|<q$, the candidate set $A$ is discarded. For each vertex $v \in T$, the algorithm forms an \spwedge whose type $x_{\kappa}$ is generated following the sign of the edges $(v, w_i)$ with respect to the sign of $(u, v)$ as described in Section~\ref{sec:wedge}. Finally, using Lemma~\ref{lemma-1}, we count the number of  balanced $(p, q)$-bicliques.

\section{Experimental Evaluation}
\label{Sec:EE}
% In this section, we evaluate the performance of the proposed wedge-centric and
% triplet-centric algorithms and compare them against the baseline algorithm
% SBCList++.
% \subsection{Algorithms in Our Experiments}
We evaluated the following algorithms.

\begin{itemize}
  \item \textbf{SBCList++:} A baseline algorithm adapted from \textbf{BCList++}~\cite{yang2023p}
  for counting balanced $(p,q)$-bicliques in signed bipartite graphs.

  \item \textbf{BBWC:} A proposed wedge-centric algorithm for counting balanced $(p,q)$-bicliques,
  as described in Section~\ref{sec:wedge}.

  \item \textbf{BBVP:} A proposed vertex-based pruning algorithm for counting balanced $(p,q)$-bicliques,
  as described in Section~\ref{sec:triplet}.

  %     \item \textbf{P-BBTC:} Multi-core based parallel implementation of BBTC counting algorithm discussed in,
  % as described in Section~\ref{sec:triplet33par1}.
\end{itemize}

\textbf{Computing Resources:} We implemented all the algorithms in C++. Experiments were conducted on a workstation with an Intel Xeon E5-1630 v4 CPU (4 cores, 8 threads, 3.70~GHz), 32~GB RAM, running Ubuntu 64-bit. 

% We utilize the \texttt{Intel TBB} library to enable parallelism in our implementation.

\begin{table}[!ht]
\footnotesize
\renewcommand{\arraystretch}{0.9}
\caption{Characteristics of the datasets and balanced $(3,3)$-biclique counts.}
\label{table:datasets_disc}
\centering
\begin{tabular}{|l|r|r|r|r|}
\hline
\textbf{Dataset} 
& \textbf{$|U|$} 
& \textbf{$|V|$} 
& \textbf{$|E|$} 
& \textbf{\# Balanced $(3,3)$-bicliques} \\
\hline
Drug\_Human (\textbf{DH})
& $820$ & $1{,}315$ & $4{,}128$ & $92{,}818$ \\
\hline
Senate (\textbf{SE})
& $145$ & $1{,}056$ & $27{,}083$ & $1{,}261{,}215{,}333$ \\
\hline
% DBLP (\textbf{DBLP})
% & $6{,}001$ & $1{,}308$ & $29{,}256$ & $2{,}204{,}537$ \\
% \hline
% Avitos (\textbf{AV})
% & $27{,}736$ & $16{,}589$ & $67{,}028$ & $586{,}475$ \\
% \hline
House (\textbf{HO})
& $515$ & $1{,}281$ & $114{,}378$ & $101{,}165{,}915{,}954$ \\
\hline
Astro (\textbf{AS})
& $13{,}643$ & $18{,}446$ & $198{,}050$ & $52{,}586{,}996$ \\
\hline
% Yelp (\textbf{Yelp})
% & $5{,}597$ & $31{,}512$ & $279{,}502$ & $41{,}586{,}181$ \\
% \hline
NIPS Papers (\textbf{NIPS})
& $1{,}500$ & $12{,}375$ & $746{,}315$ & $964.03\,\mathrm{B}$ \\
\hline
Jester 150 (\textbf{JE})
& $50{,}692$ & $140$ & $1{,}728{,}847$ & $1.00\,\mathrm{T}$ \\
\hline
KDD Cup (\textbf{KDD})
& $255{,}170$ & $1{,}848{,}114$ & $2{,}766{,}393$ & $10{,}621{,}521$ \\
\hline
AOL (\textbf{AOL})
& $4{,}811{,}647$ & $1{,}632{,}788$ & $10{,}741{,}953$ & $83{,}736{,}979$ \\
\hline
Epinions (\textbf{EP})
& $120{,}492$ & $755{,}760$ & $13{,}668{,}320$ & $497.95\,\mathrm{T}$ \\
\hline
\end{tabular}
\end{table}

\textbf{Datasets Description}: We evaluate our proposed algorithms on a diverse collection of $9$ real-world bipartite datasets. The \textbf{DH} dataset is obtained from~\cite{torres2016drug}, while \textbf{SE} and \textbf{HO} are collected from a signed bipartite network repository\footnote{\url{https://github.com/tylersnetwork/signed_bipartite_networks}}. The datasets \textbf{KDD}, and \textbf{AOL} are sourced from the bipartite network repository\footnote{\url{https://renchi.ac.cn/datasets/}}, and the remaining datasets, namely \textbf{AS}, \textbf{NIPS}, \textbf{JE}, and \textbf{EP}, are obtained from the KONECT collection\footnote{\url{http://konect.cc/networks/}}. 

Among the datasets, DH, SE, and HO are natively provided as signed bipartite graphs. Some datasets contain explicit rating information, which we convert into signed bipartite graphs by binarizing the ratings. Specifically, in the Jester-150 (JE) dataset, which uses a $10$-star rating system, ratings greater than $6$ are assigned a positive sign $(1)$, while all remaining ratings are assigned a negative sign $(0)$. Similarly, in the EP dataset, which uses a $5$-star rating system, ratings of $4$ or higher are labeled as positive $(1)$, and ratings of $3$ or lower are labeled as negative $(0)$. The remaining datasets (AS, NIPS, KDD, and AOL) are originally unsigned and are converted into synthetic signed bipartite graphs following the methodology in~\cite{chen2021maximum}. In this process, each edge is assigned a positive label $(1)$ with probability $0.7$ and a negative label $(0)$ with probability $0.3$. Table~\ref{table:datasets_disc} summarizes the characteristics of the datasets, including the total number of balanced $(3,3)$-bicliques; for extremely large counts, we report values in billions (B) and trillions (T).

\subsection{Performance Assessment}
We evaluate the performance of the proposed algorithms on the datasets described earlier and compare them with the baseline algorithm in terms of processing time and memory usage. We limit the maximum running time of each program to $5$ hours; if an execution exceeds this limit, we report the result as \texttt{INF}.

\textbf{Processing Time}: We first present the results for balanced $(p,q)$-bicliques with $p=q=3$, as shown in Table~\ref{tab:runtime_memory}(a). We begin by comparing the baseline method SBCList++ with BBWC, our first proposed algorithm. BBWC shows substantial improvement across all datasets, with the baseline finishing within the time limit. For example, on the DH dataset, BBWC reduces the runtime from $0.755\,s$ to $0.05\,s$, achieving a speedup of approximately $15.1\times$ (about 93\% reduction in runtime). On the SE dataset, the runtime decreases from $1937\,s$ to $3.77\,s$, corresponding to a speedup of about $561\times$. Similarly, on the AS dataset, BBWC achieves a $11.2\times$ speedup, and on the KDD dataset, a $2.8\times$ speedup. For the remaining datasets (HO, NIPS, and JE), the baseline algorithm SBCList++ exceeds the imposed time limit of $5$ hours; in contrast, BBWC completes within the time bound, demonstrating significantly better scalability. These improvements arise from fundamental algorithmic differences. Although SBCList++ correctly counts balanced bicliques, it performs redundant computations on many unbalanced biclique candidates. In the worst case, it may enumerate many $(3,3)$-bicliques that do not satisfy the balance constraint, leading to substantial computational overhead and increased execution time. However, we observe that although BBWC outperforms the baseline in general, its performance degrades on larger datasets such as AOL and EP. The main bottleneck in BBWC is wedge enumeration, which dominates runtime and reduces pruning effectiveness. This issue is particularly pronounced on datasets with larger minimum partition sizes, where the number of signed 3-wedges enumerated grows significantly. As a result, BBWC may fail to complete within the imposed time limit on such datasets. To address this, we propose BBVP, which reduces the overhead of signed wedge enumeration and improves pruning efficiency.

  \begin{table}[!htbp]
\centering
\setlength{\tabcolsep}{3pt}
\caption{Runtime and Memory usage comparison of proposed algorithms on various datasets with (3,3) parameter setting.}
\label{tab:runtime_memory}

% ================= Runtime =================
\begin{minipage}{0.44\textwidth}
\centering
\setlength{\tabcolsep}{4pt}
\renewcommand{\arraystretch}{0.85}
\textbf{(a) Processing Time (seconds)}\\[2pt]
\begin{tabular}{|l|l|l|l|}
\hline
\textbf{Dataset} & \textbf{SBCList++} & \textbf{BBWC} & \textbf{BBVP} \\
\hline
DH   & 0.755 & 0.05 & \textbf{0.01} \\
\hline

SE & 1937 & 3.77572 & \textbf{0.81424} \\
\hline

% DBLP & 10.425 & 0.672 & \textbf{0.108} \\
% \hline
% AV   & 0.761 & 2.716 & \textbf{0.572} \\
% \hline

HO   & INF & 231.848 & \textbf{41.381} \\
\hline

AS   & 493.240 & 44.848 & \textbf{7.381} \\
\hline
% Yelp & 102 & 107 & \textbf{17} \\
% \hline
NIPS & INF & 15717 & \textbf{3840} \\
\hline
JE   & INF & 1741 & \textbf{351} \\
\hline
KDD  & 217.075 & 87 & \textbf{8.211} \\
\hline
AOL  & INF & INF & \textbf{1764} \\
\hline
EP   & INF & INF & \textbf{8346} \\
\hline
\end{tabular}
\end{minipage}
\hfill
% ================= Memory =================
\begin{minipage}{0.48\textwidth}
\centering
\setlength{\tabcolsep}{4pt}
\renewcommand{\arraystretch}{0.85}
\textbf{(b) Memory Usage (MB)}\\[2pt]
\begin{tabular}{|l|l|l|l|}
\hline
\textbf{Dataset} & \textbf{SBCList++} & \textbf{BBWC} & \textbf{BBVP} \\
\hline
DH   & 0.006 & 0.005 & \textbf{0.003} \\
\hline

SE & 0.682 & 0.008 & \textbf{0.007} \\
\hline

% DBLP & 0.018 & 0.011 & 0.013 \\
% \hline
% AV   & 0.042 & 0.019 & 0.018 \\
% \hline

HO   & INF & 0.025 & \textbf{0.020} \\
\hline

AS   & 0.193 & 0.053 & \textbf{0.030} \\
\hline
% Yelp & 0.206 & 0.081 & 0.048 \\
% \hline
NIPS & INF & 0.17 & \textbf{0.133} \\
\hline
JE   & INF & 0.331 & \textbf{0.272} \\
\hline
KDD  & 1.117 & 0.747 & \textbf{0.626} \\
\hline
AOL  & INF & INF & \textbf{2.241} \\
\hline
EP   & INF & INF & \textbf{2.312} \\
\hline
\end{tabular}
\label{tab:untime_memory}
\end{minipage}

\end{table}

\begin{figure*}[!ht]
    \centering
    \includegraphics[width=0.89\textwidth]{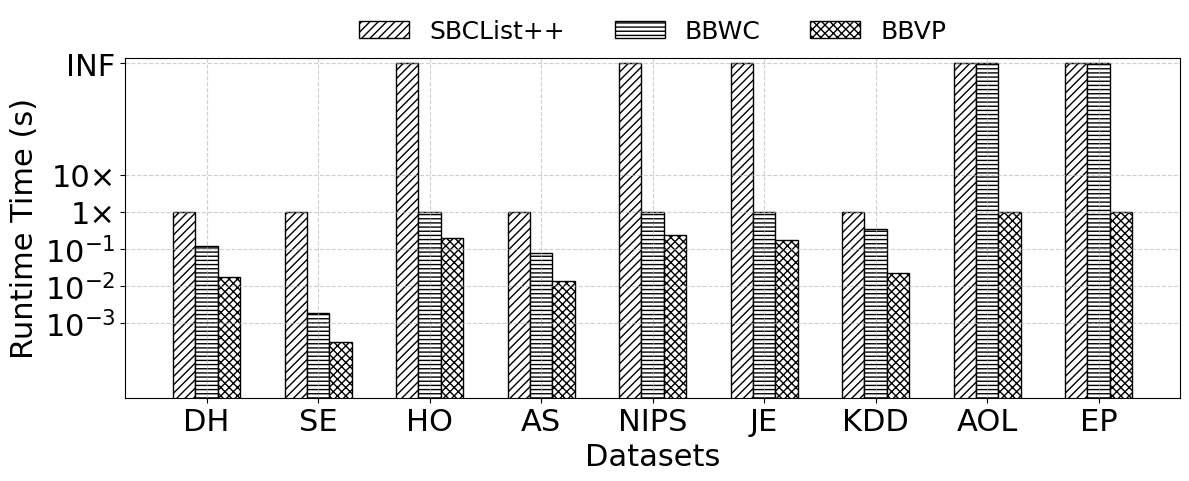}
\caption{Runtime (in seconds) of the proposed algorithms on various datasets for $(3,3)$ parameter setting.}
    \label{fig:33_runtime2}
\end{figure*}

As shown in Table~\ref{tab:runtime_memory}(a), BBVP consistently outperforms both BBWC and SBCLIst++ across all evaluated datasets. On the Datasets where BBWC already performs well, BBVP provides additional improvements, yielding an average speedup of approximately $10.7\times$ over BBWC. The gains are particularly notable on datasets with large search spaces, such as HO and AS, where BBVP achieves speedups of about $24.7\times$ and $18.8\times$, respectively. Even on smaller datasets like DH and SE, BBTC continues to reduce runtime by factors of $5\times$ and $3.8\times$. When compared with the baseline SBCList++, the improvements become even more pronounced. considering only the datasets where the baseline completes within the time limit (DH, SE, AS, and KDD), BBVP achieves an average speedup of approximately $561\times$. Furthermore, BBVP successfully completes on the largest datasets, such as AOL and EP, where both SBCList++ and BBWC exceed the $5$-hour time limit. These results demonstrate the strong scalability of BBVP for higher-order balanced biclique counting. Fig.~\ref{fig:33_runtime2} summarizes the runtime performance of the proposed algorithms compared with the baseline across all datasets. We report \textit{normalized runtime} to enable fair comparison. For each dataset, runtimes are divided by the SBCList++ runtime (when available). Thus, a value of $1\times$ denotes equal performance to SBCList++, values below $1$ indicate faster execution, and values above $1$ indicate slower execution. When SBCList++ does not finish, runtimes are normalized with respect to the slowest completed method on that dataset.

\begin{figure*}[!ht]
    \centering
    \includegraphics[width=0.85\textwidth]{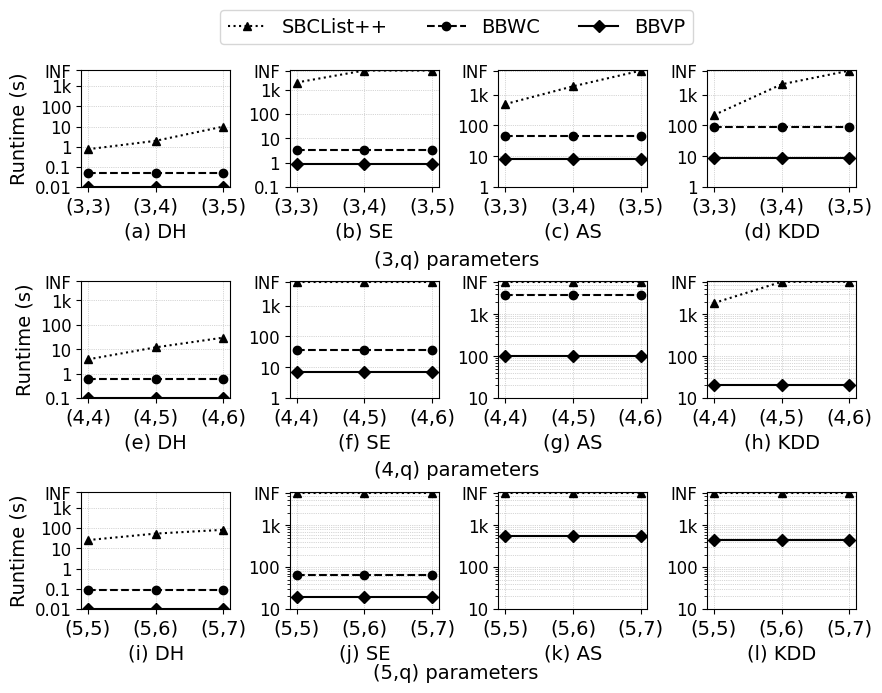}
    \caption{Running time (s) of the proposed algorithms for different $(p,q)$ parameter settings.}
    \label{fig:3k4k5kruntime1}
\end{figure*}

\begin{figure*}[!ht]
    \centering
    \includegraphics[width=0.85\textwidth]{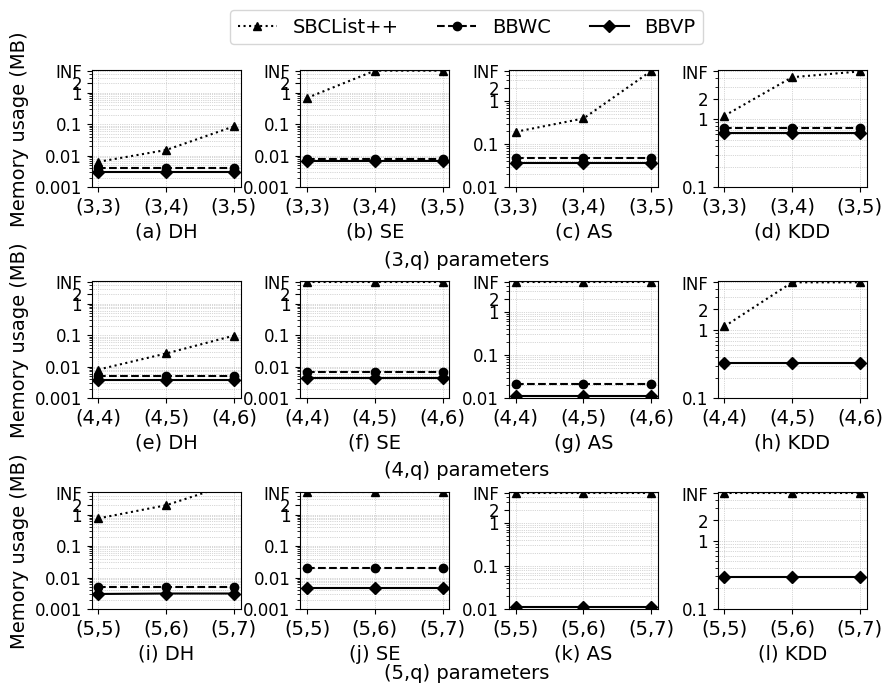}
    \caption{Memory usage (MB) of the proposed algorithms for different $(p,q)$ parameter settings.}
    \label{fig:3k4k5kmemtime}
\end{figure*}

Next, we extend the evaluation to more general balanced $(p,q)$-bicliques with $p \in \{3,4,5\}$ and $q \in \{3,4,5,6,7\}$. This is particularly relevant in motif analysis as motifs are generally small subgraphs that play significant roles in network analysis~\cite{li2021mixed,niu2025fast}. To ensure a fair comparison, we report results only for the datasets where the baseline method completes within the imposed time limit, namely DH, SE, AS, and KDD. Additionally, we observe that the baseline's performance deteriorates as $q$ increases. In contrast, the BBWC and BBVP algorithms maintain consistent performance as they have already separated the signed wedges that form balanced bicliques and choose the q combinations from them. Fig.~\ref{fig:3k4k5kruntime1} demonstrates the running time (s) of the proposed algorithms for different $(p,q)$ parameter settings.

\textbf{Memory Usage}: This subsection evaluates the memory consumption of the proposed algorithms across different datasets and parameter settings. Table~\ref{tab:runtime_memory}(b) reports the memory usage of all algorithms for counting balanced $(3,3)$-bicliques. The results demonstrate that the proposed methods, BBWC and BBVP, require less memory than the baseline SBCList++ across all datasets. The higher memory usage of SBCList++ is due to its reliance on additional auxiliary arrays to avoid frequent subgraph constructions. A comparison between the proposed methods further shows that BBVP is generally more memory-efficient than BBWC. This is because BBVP eliminates many signed wedges at an early stage if they cannot participate in a balanced biclique, thereby reducing intermediate storage requirements. Additionally, we observe that the baseline's memory usage is increasing as $q$ increases. In contrast, the BBWC and BBVP algorithms maintain consistent performance. Fig.~\ref{fig:3k4k5kmemtime} summarizes the memory usage of the proposed algorithms compared with the baseline across all datasets for different $(p,q)$ parameter settings.

% \begin{figure*}[!ht]
%     \centering
%     \includegraphics[width=0.99\textwidth]{Images/33_memory8.png}
%     \caption{Memory usage (MB) of the proposed algorithm on various datasets.}
%     \label{fig:33_memory}
% \end{figure*}

% Table~\ref{tab:runtime_memory}(b) reports the memory consumption of all algorithms for counting balanced $(3,3)$-bicliques. The proposed methods, BBWC and BBTC, exhibit memory usage comparable to or lower than the baseline SBCList++ across most datasets.  On datasets where SBCList++ fails to terminate (e.g., HO, NIPS, JE, AOL, and EP), memory usage is reported as INF in Table~\ref{tab:runtime_memory}(b). In contrast, both BBWC and BBTC maintain bounded memory consumption. While BBTC incurs slightly higher memory usage than BBWC on some datasets due to additional triplet-level bookkeeping, its memory footprint remains modest. Overall, the proposed algorithms achieve significant runtime improvements without incurring excessive memory overhead.

% Overall, Fig.~\ref{fig:33_runtime2} and Fig.~\ref{fig:33_memory} summarise the runtime performance and memory consumption of the proposed algorithms compared with the baseline across all datasets.

% \begin{figure*}[!ht]
%     \centering
%     \includegraphics[width=0.90\textwidth]{Images/33_runtime4.png}
% \caption{Runtime (in seconds) of the proposed algorithms on various datasets.}
%     \label{fig:33_runtime2}
% \end{figure*}

% \vspace{-1.3cm}

\section{Conclusion and Remarks}
\label{Sec:Con}

In this work, we introduced the problem of balanced $(p,q)$-biclique counting in signed bipartite graphs and proposed efficient wedge-centric and vertex-based pruning algorithms. 
Extensive experimental evaluation across multiple real-world datasets demonstrates that the proposed methods consistently outperform the adapted baseline, SBCList++, in both runtime and memory usage. An important direction for future work is to extend the proposed algorithms to multi-core and GPU architectures in order to further improve scalability. 
Another challenging and practically relevant direction is to adapt the proposed techniques to a dynamic signed bipartite graph.

\bibliography{lipics-v2021-sample-article}

\begin{thebibliography}{10}

\bibitem{aksoy2017measuring}
Sinan~G Aksoy, Tamara~G Kolda, and Ali Pinar.
\newblock Measuring and modeling bipartite graphs with community structure.
\newblock {\em Journal of Complex Networks}, 5(4):581--603, 2017.

\bibitem{chen2021maximum}
Chen Chen, Yanping Wu, Renjie Sun, and Xiaoyang Wang.
\newblock Maximum signed $\theta$-clique identification in large signed graphs.
\newblock {\em IEEE Transactions on Knowledge and Data Engineering}, 35(2):1791--1802, 2021.

\bibitem{chen2022efficient}
Lu~Chen, Chengfei Liu, Rui Zhou, Jiajie Xu, and Jianxin Li.
\newblock Efficient maximal biclique enumeration for large sparse bipartite graphs.
\newblock {\em Proceedings of the VLDB Endowment}, 15(8):1559--1571, 2022.

\bibitem{chen2021higher}
Zi~Chen, Long Yuan, Li~Han, and Zhengping Qian.
\newblock Higher-order truss decomposition in graphs.
\newblock {\em IEEE Transactions on Knowledge and Data Engineering}, 35(4):3966--3978, 2021.

\bibitem{chen2020efficient}
Zi~Chen, Long Yuan, Xuemin Lin, Lu~Qin, and Jianye Yang.
\newblock Efficient maximal balanced clique enumeration in signed networks.
\newblock In {\em Proceedings of The Web Conference 2020}, pages 339--349, 2020.

\bibitem{chung2023maximum}
Kai~Hiu Chung, Alexander Zhou, Yue Wang, and Lei Chen.
\newblock Maximum balanced (k, $\epsilon$)-bitruss detection in signed bipartite graph.
\newblock {\em Proceedings of the VLDB Endowment}, 17(3):332--344, 2023.
\newblock \href {https://doi.org/10.14778/3632093.3632099} {\path{doi:10.14778/3632093.3632099}}.

\bibitem{derr2019balance}
Tyler Derr, Cassidy Johnson, Yi~Chang, and Jiliang Tang.
\newblock Balance in signed bipartite networks.
\newblock In {\em Proceedings of the 28th ACM International Conference on Information and Knowledge Management}, pages 1221--1230, 2019.

\bibitem{heider1946attitudes}
Fritz Heider.
\newblock Attitudes and cognitive organization.
\newblock {\em The Journal of psychology}, 21(1):107--112, 1946.

\bibitem{kiran2024efficient}
Mekala Kiran, Apurba Das, and Suman Banerjee.
\newblock Efficient counting of balanced (2, k)-bicliques in signed bipartite graphs.
\newblock In {\em 2024 IEEE International Conference on Big Data (BigData)}, pages 735--740. IEEE, 2024.

\bibitem{li2021mixed}
Chao Li, Qiming Yang, Bowen Pang, Tiance Chen, Qian Cheng, and Jiaomin Liu.
\newblock A mixed strategy of higher-order structure for link prediction problem on bipartite graphs.
\newblock {\em Mathematics}, 9(24):3195, 2021.

\bibitem{luo2023efficient}
Wensheng Luo, Qiaoyuan Yang, Yixiang Fang, and Xu~Zhou.
\newblock Efficient core maintenance in large bipartite graphs.
\newblock {\em Proceedings of the ACM on Management of Data}, 1(3):1--26, 2023.

\bibitem{lyu2020maximum}
Bingqing Lyu, Lu~Qin, Xuemin Lin, Ying Zhang, Zhengping Qian, and Jingren Zhou.
\newblock Maximum biclique search at billion scale.
\newblock {\em Proceedings of the VLDB Endowment}, 2020.

\bibitem{mitzenmacher2015scalable}
Michael Mitzenmacher, Jakub Pachocki, Richard Peng, Charalampos Tsourakakis, and Shen~Chen Xu.
\newblock Scalable large near-clique detection in large-scale networks via sampling.
\newblock In {\em Proceedings of the 21th ACM SIGKDD International Conference on Knowledge Discovery and Data Mining}, pages 815--824, 2015.

\bibitem{niu2025fast}
Jason Niu, Jaroslaw Zola, and Ahmet~Erdem Sar{\i}y{\"u}ce.
\newblock Fast counting and utilizing induced 6-cycles in bipartite networks.
\newblock {\em IEEE Transactions on Knowledge and Data Engineering}, 2025.

\bibitem{opsahl2013triadic}
Tore Opsahl.
\newblock Triadic closure in two-mode networks: Redefining the global and local clustering coefficients.
\newblock {\em Social networks}, 35(2):159--167, 2013.

\bibitem{sanei2018butterfly}
Seyed-Vahid Sanei-Mehri, Ahmet~Erdem Sariyuce, and Srikanta Tirthapura.
\newblock Butterfly counting in bipartite networks.
\newblock In {\em Proceedings of the 24th ACM SIGKDD International Conference on Knowledge Discovery \& Data Mining}, pages 2150--2159, 2018.

\bibitem{shi2020parallel}
J.~Shi and J.~Shun.
\newblock Parallel algorithms for butterfly computations.
\newblock In {\em Symposium on Algorithmic Principles of Computer Systems (APoCS)}, pages 16--30, 2020.

\bibitem{sun2022maximal}
Renjie Sun, Yanping Wu, Chen Chen, Xiaoyang Wang, Wenjie Zhang, and Xuemin Lin.
\newblock Maximal balanced signed biclique enumeration in signed bipartite graphs.
\newblock In {\em 2022 IEEE 38th International Conference on Data Engineering (ICDE)}, pages 1887--1899. IEEE, 2022.

\bibitem{sun2023efficient}
Renjie Sun, Yanping Wu, Xiaoyang Wang, Chen Chen, Wenjie Zhang, and Xuemin Lin.
\newblock Efficient balanced signed biclique search in signed bipartite graphs.
\newblock {\em IEEE Transactions on Knowledge and Data Engineering}, 2023.

\bibitem{sun2020discovering}
Renjie Sun, Qiuyu Zhu, Chen Chen, Xiaoyang Wang, Ying Zhang, and Xun Wang.
\newblock Discovering cliques in signed networks based on balance theory.
\newblock In {\em International Conference on Database Systems for Advanced Applications}, pages 666--674. Springer, 2020.

\bibitem{torres2016drug}
N{\'u}ria~Ballber Torres and Claudio Altafini.
\newblock Drug combinatorics and side effect estimation on the signed human drug-target network.
\newblock {\em BMC systems biology}, 10:1--12, 2016.

\bibitem{wang2014rectangle}
Jia Wang, Ada Wai-Chee Fu, and James Cheng.
\newblock Rectangle counting in large bipartite graphs.
\newblock In {\em 2014 IEEE International Congress on Big Data}, pages 17--24. IEEE, 2014.

\bibitem{wang2024efficient}
Jianhua Wang, Jianye Yang, Zhaoquan Gu, Dian Ouyang, Zhihong Tian, and Xuemin Lin.
\newblock Efficient maximal biclique enumeration on large signed bipartite graphs.
\newblock {\em IEEE Transactions on Knowledge and Data Engineering}, 36(9):4618--4631, 2024.

\bibitem{wang2019vertex}
Kai Wang, Xuemin Lin, Lu~Qin, Wenjie Zhang, and Ying Zhang.
\newblock Vertex priority based butterfly counting for large-scale bipartite networks.
\newblock {\em PVLDB}, 2019.

\bibitem{xia2024gpu}
Y.~Xia, F.~Zhang, Q.~Xu, M.~Zhang, Z.~Yao, L.~Lu, X.~Du, D.~Deng, B.~He, and S.~Ma.
\newblock Gpu-based butterfly counting.
\newblock {\em The VLDB Journal}, 33(5):1543--1567, 2024.

\bibitem{yang2023p}
Jianye Yang, Yun Peng, Dian Ouyang, Wenjie Zhang, Xuemin Lin, and Xiang Zhao.
\newblock (p, q)-biclique counting and enumeration for large sparse bipartite graphs.
\newblock {\em The VLDB Journal}, 32(5):1137--1161, 2023.

\bibitem{yang2021efficient}
Yixing Yang, Yixiang Fang, Maria~E Orlowska, Wenjie Zhang, and Xuemin Lin.
\newblock Efficient bi-triangle counting for large bipartite networks.
\newblock {\em Proceedings of the VLDB Endowment}, 14(6):984--996, 2021.

\bibitem{yao2022identifying}
Kai Yao, Lijun Chang, and Jeffrey~Xu Yu.
\newblock Identifying similar-bicliques in bipartite graphs.
\newblock {\em Proceedings of the VLDB Endowment}, 15(11):3085--3097, 2022.

\bibitem{zou2016bitruss}
Zhaonian Zou.
\newblock Bitruss decomposition of bipartite graphs.
\newblock In {\em International conference on database systems for advanced applications}, pages 218--233. Springer, 2016.

\end{thebibliography}
\section{Appendix}
\appendix

\subsection{Structural property of Balanced Bicliques}
\label{app:spwedge-type-proof}
\begin{lemma}
    All \spwedge{s} in a balanced $(p, q)$-biclique are of the same type.
\end{lemma}

\begin{proof}
    We will prove this by contradiction. Let there is a balanced $(p, q)$-biclique $b$ where there are two different types of \pwedge $\mathbf{x}^p_{\kappa} = \langle u, w_1, w_2, ..., w_{p-1}, v\rangle$ and $\mathbf{x}^p_{\kappa'} = \langle u, w_1, w_2, ..., w_{p-1}, v'\rangle$. Suppose in the representations of $\kappa = c_1c_2...c_{p-1}$ and $\kappa'=c_1'c_2'...c_{p-1}'$, $i$ is the first index where $c_i\neq c_i'$. Without loss of generality, assume that $c_i=s$ and $c_i'=d$. Then, the butterfly $(u, w_i, v, v')$ is an unbalanced butterfly. This is a contradiction.
% \hfill $\square$
\end{proof}

\subsection{Time complexity Analysis of Algorithm~\ref{algo:bucketpq}}
\label{app:complexity-proof}

\begin{theorem}
    The worst case time complexity of Algorithm~\ref{algo:bucketpq} for counting \bpqbc in a signed bipartite graph $G = (U, V, E)$ with $|U|=m, |V|=n$ is $O(m\Delta\binom{\Delta}{p})$ where $\Delta$ is the maximum degree of $G$.
\end{theorem}

\begin{proof}
The time complexity has two parts. First, we compute the complexity of counting the number of vertices connected to each \spwedge in the respective bucket. For this, we construct a vertex subset of size $p$ within the $2$-hop neighborhood of a vertex $u$. The worst-case complexity of this counting is $\Delta\times\binom{\Delta}{p}$ as $\Delta$ is the maximum degree in the graph. For choosing a specific \spwedge respecting the vertex order, we can sort the vertices using counting sort in time $O(m)$. In the second phase of the algorithm, we count the number of \bpqbc using the information of the number of common vertices for each \spwedge (Lines $8$-$10$). For this, we require $\Delta\times\binom{\Delta}{p}$ entries across the buckets, since we have processed these many \spwedge{s} in the first phase of the algorithm. Thus, overall time complexity becomes $O(m\Delta\binom{\Delta}{p})$.
\end{proof}

\end{document}